\documentclass[Journal,twocolumn,10pt]{IEEEtran}
\usepackage{graphicx}
\usepackage{amssymb}
\usepackage{amsmath}

\usepackage{algorithm}
\usepackage{algpseudocode}
\usepackage{algorithmicx}

\newcommand{\comment}[1]{}

\newtheorem{theorem}{Theorem}[section]
     \newtheorem{lemma}[theorem]{Lemma}

     \newenvironment{definition}[1][Definition:]{\begin{trivlist}
     \item[\hskip \labelsep {\bfseries #1}]}{\end{trivlist}}

     \newcommand{\qed}{\nobreak \ifvmode \relax \else
           \ifdim\lastskip<1.5em \hskip-\lastskip
           \hskip1.5em plus0em minus0.5em \fi \nobreak
           \vrule height0.75em width0.5em depth0.25em\fi}
\begin{document}

% paper title

\title{Joint Write-Once-Memory and Error-Control Codes}

\author{\authorblockN{Xudong Ma \\}
\authorblockA{Pattern Technology Lab LLC, U.S.A.\\
Email: xma@ieee.org} }

\maketitle

\begin{abstract}

Write-Once-Memory (WOM) is a model for many modern non-volatile memories, such as flash memories. Recently, several capacity-achieving WOM coding schemes have been proposed based on polar coding. Due to the fact that practical computer memory systems always contain noises, a nature question to ask next is how may we generalize these coding schemes, such that they may also have the error-control capabilities. In this paper, we discuss a joint WOM and error-control coding scheme, which is a generalization of the capacity-achieving WOM codes based on source polarization in \cite{ma15}. In this paper, we prove a sufficient and necessary condition for the noisy reading channel being less noisy than the test channel in data encoding in the polar WOM coding. Such a sufficient and necessary condition is usually satisfied in reality. As a consequence of the sufficient and necessary condition, the high entropy set related to the noisy channel is usually strictly contained in the high entropy set related to the test channel in data encoding. Therefore the low-complexity polar joint WOM and error-control codes are sufficient for most practical coding scenarios.

 
%\keywords{polar code, write-oce-memory, information theory} 

\end{abstract}

\section{Introduction}

\label{section_introduction}

Write-Once-Memories (WOM) are a frequently used model for modern computer memories, such as flash memories. For example, for Single-Level-Cell (SLC) type flash memories, each memory cell may be turned from bit 1 to bit 0 by applying programming pulses. However, the memory cell can not be turned from bit 0 to bit 1 without using a block erasure operation. If no erasure operation is used, then the SLC type flash memory is a perfect example of WOM.

Modern NAND-type flash memories usually can endure very limited numbers of erasure operations. The block erasure operations are also expensive in terms of time and power consumption  and thus should be avoided as much as possible. Rewriting coding schemes for WOM are a type of coding schemes, where data can still be written into WOM, when some memory cells have already been written. Rewriting codes for flash memories are thus preferred, due to the fact that data can be written into each memory block multiple times without using any erasure operation.

The concept of rewriting codes for WOM was coined by Rivest and Shamir \cite{rivest82}. However, constructing capacity-achieving WOM codes had been an open problem for almost three decades. Recently, several capacity-achieving WOM codes were proposed. One algebraic construction based capacity-achieving WOM coding scheme is proposed by Shipilka \cite{shipilka12}. The coding scheme achieves zero encoding error probability, but have polynomial computational complexities. Another polar coding based capacity-achieving WOM coding scheme is proposed by Burshtein and Strugatski~\cite{burshtein12}. This coding scheme has asymptotic vanishing encoding error probability and linear coding complexities. However, the coding scheme needs a random dither signal shared by the encoder and decoder. For most data storage applications, the random dither signal is difficult to implement. A third capacity-achieving WOM code construction was proposed in~\cite{ma15}, where the coding scheme is based on source polarization~\cite{arikan10}. The WOM coding scheme in~\cite{ma15} has vanishing encoding error probabilities and linear coding complexities. One major advantage of this coding scheme is that no shared random dither signal is needed.

Due to the fact that practical computer memory systems always contain noises, a nature question to ask next is how may we generalize the above coding schemes, such that they may also have the error-control capabilities. Some joint WOM and channel coding schemes have already been proposed in~\cite{gad14}. The coding schemes in \cite{gad14} may be considered as generalizations of the WOM coding schemes in \cite{burshtein12}.

Recently, some new joint WOM and channel coding schemes were also proposed in \cite{gad14b}. In \cite{gad14b}, two coding schemes were proposed including one basic coding scheme, and one chaining based coding scheme. The basic coding scheme may be considered as a certain generalization of the WOM coding scheme in \cite{ma15} from the noiseless WOM case to the noisy WOM case. However, the paper \cite{gad14b} only showed that the basic coding scheme can be used for a specific channel model, where any memory cells with bit 0 can be read with zero error probability. Such a channel model is certainly not practical in reality. Therefore, a more sophisticated coding scheme was proposed in \cite{gad14b} for more general cases. The sophisticated coding scheme is built upon the basic coding scheme and chaining, a coding strategy previously used mainly for broadcast channels.

In this paper, we discuss a joint WOM and channel coding scheme similar to the basic coding scheme in \cite{gad14b}. As one main contribution of this paper, we prove a sufficient and necessary condition for the noisy reading channel being less noisy than the test channel in data encoding. Such a sufficient and necessary condition is usually satisfied in practical scenarios, as long as the noisy reading process of the WOM is not extremely noisy. We prove that if the sufficient and necessary condition is satisfied, then in the polarization process, the high entropy set due to the noisy channel is always contained in the high entropy set due to the test channel for data encoding. Thus, the very simple coding scheme is sufficient for most practical WOM coding scenarios.

The rest of the paper is organized as follows. In Section \ref{section_basic_scheme}, we discuss the considered joint WOM and error-control coding scheme.  In Section \ref{section_less_noisy}, we present our proof of the sufficient and necessary condition for less noisy. In the section, we show that the high entropy set due to the noisy channel is always contained in the high entropy set due to the test channel for data encoding, if the sufficient and necessary condition is satisfied. Finally, some concluding remarks are presented in Section \ref{sec_conclusion}.

We use the following notation throughout this paper. We  use $X_{n}^{N}$ to denote a sequence of symbols $X_n, X_{n+1}, \ldots, X_{N-1}, X_{N}$. With some abuse of notation, we may also use $X_{n}^{N}$ to denote a row vector $[X_n,X_{n+1},\ldots,X_N]$. We use upper-case letters to denote random variables and lower-case letters to denote the corresponding realizations. For example, $X_i$ is one random variable, and $x_i$ is one realization of the random variable $X_i$. We use $H(X)$ to denote the entropy of the random variable $X$, $I(X;Y)$ to denote the mutual information between the random variables $X$, $Y$, and many other standard information-theoretic notation used previously (for example as in \cite{cover06}). If $\mathcal{G}$ is a set, then we use $\mathcal{G}^c$ to denote the complement of $\mathcal{G}$.

\section{The joint WOM and error-control coding scheme}
\label{section_basic_scheme}

In this paper, our discussions are limited to the binary WOM case. However, it should be clear that the discussions in this paper may be generalized easily to non-binary cases using non-binary polar coding, for example the coding schemes in \cite{mori10} and \cite{park13}. The corresponding considered flash memory is thus of SLC type. We consider a SLC type flash memory block with $N$ memory cells. We use the random variables $S_{1}^{N}$ to denote the memory cell states before the considered round of data encoding. We use the random variable $X_{1}^{N}$ to denote the codeword during the current round of data encoding. We use $Y_{1}^{N}$ to denote the noisy sensing results of the codeword $X_{1}^{N}$. We assume that each $Y_i$ is a discrete random variable with a finite size alphabet. The assumption is usually true for the case of flash memories, because there are always finitely many sensing threshold voltage levels. However, it should be noted that the results in this paper can be generalized to other WOM coding cases, where each $Y_i$ is a continuous random variable. The generalization can be achieved by considering fine quantization of the continuous random variable $Y_i$.

As a quick example, if we assume that the memory block contains 4 memory cells, then $N=4$. The memory cell state $S_{1}^{N} = [1,1,0,0]$ implies that the first memory cell is at bit 1, the second memory cell is at bit 1, the third memory cell is at bit 0, and the fourth memory cell is at bit 0. The codeword $X_{1}^{N}$ may be $[0,1,0,0]$. In this case, the first memory cell is turned from bit 1 to bit 0, in order to record a certain message.

We define the polar transformation as $U_{1}^{N}=X_{1}^{N}G_k$. The matrix $G_k$ is defined as
\begin{align}
G_k = \left[\begin{array}{cc}
1 & 0 \\
1 & 1
\end{array}\right]^{\otimes k}
\end{align}
where, $\otimes k$ denotes the k-th Kronecker power, $k=\log_{2} N$.
Note the in the traditional definitions,  in the channel polarization process $U_{1}^{N}=X_{1}^{N}{G_k}^{-1}$, and in the source polarization process $U_{1}^{N}=X_{1}^{N}{G_k}$. There are also some bit-reversal permutations in the transforms. However, for the binary case ${G_k}^{-1} = G_k$. Therefore, the definition in the current paper is roughly equal to the traditional definitions and is frequently used in the recent research papers. Note that generally speaking  ${G_k}^{-1} = G_k$ may not hold for non-binary polar coding cases. Thus, extra cares need to be taken, if necessary.

In this section, we  present the  encoding and decoding algorithms of the discussed joint WOM and error-control codes. We assume that the states of memory cells $S_i$ are independent and identically distributed  (iid) with the following probability distribution $P(\cdot)$,
\begin{align}
\label{model_p_2}
{P}(S_i) = \left\{
\begin{array}{ll} \beta, & \mbox{if } S_i = 0 \\
1-\beta, & \mbox{if }S_i = 1
\end{array}
\right.
\end{align}
where, $\beta$ is a real number, $0<\beta<1$.
Let each $X_i$ be conditionally independent of the other $X_j$ given $S_i$.
\begin{align}
\label{model_p_3}
{P}(X_i|S_i) = \left\{
\begin{array}{ll}
1,  & \mbox{for } X_i=0, S_i=0 \\
0,  & \mbox{for } X_i=1, S_i=0 \\
\gamma,  & \mbox{for } X_i=0, S_i=1 \\
1-\gamma,  & \mbox{for } X_i=1, S_i=1
\end{array}
\right.
\end{align}
where, $\gamma$ is a real number, $0<\gamma<1$.

From the source polarization theory established in \cite{arikan10}, we have that there exists a set ${\mathcal F}_N$, such that for each $i\in {\mathcal F}_N$, $H(U_i|S_{1}^{N}, U_{1}^{i-1})$ is close to $1$, and for each $i\notin {\mathcal F}_N$, $H(U_i|S_{1}^{N}, U_{1}^{i-1})$ is close to $0$. We call the set ${\mathcal F}_N$ as the high-entropy set corresponding to the test channel of data encoding. Similarly, we have that there exists a set ${\mathcal G}_N$, such that for each $i\in {\mathcal G}_N$, $H(U_i|Y_{1}^{N}, U_{1}^{i-1})$ is close to $1$, and for each $i\notin {\mathcal G}_N$, $H(U_i|Y_{1}^{N}, U_{1}^{i-1})$ is close to $0$. We call the set ${\mathcal G}_N$ as the high-entropy set corresponding to the noisy channel $P(Y_i|X_i)$.

In Theorem \ref{entropy_set_theorem}, we will prove a sufficient and necessary condition for having ${\mathcal G}_N \subset {\mathcal F}_N$. Such a condition is usually satisfied for most practical scenarios. Therefore, a very simple joint WOM and error-control coding scheme as shown in Algorithms \ref{polar_encoding_algorithm} and \ref{polar_decoding_algorithm} can be usually applied.

\begin{algorithm}
\caption{Joint WOM  and error-control encoding}\label{polar_encoding_algorithm}
\begin{algorithmic}[1]
\State The algorithm takes inputs
\begin{itemize}
\item the current memory states of the memory cells $s_1^N$
\item the high-entropy set ${\mathcal F}_N$
\item the high-entropy set ${\mathcal G}_N$
\item the to-be-recorded message $v_1^M$
\item freeze bits $f_{1}^{L}$
\end{itemize}
\State{$n\gets 1, m\gets 1, l\gets 1$}
\Repeat
  \If{ $n\in F_N$}

    \If{ $n\in G_N$}

      \State $u_n \gets f_l$
      \State $n\gets n+1$
      \State $l\gets l+1$

  \Else

    \State $u_n \gets v_m$
    \State $n\gets n+1$
    \State $m\gets m+1$

  \EndIf

  \Else
    \State Calculate $P(U_n|y_1^N,u_{1}^{n-1})$
    \State Randomly set $u_n$ according to the probability distribution $P(U_n|s_1^N,u_{1}^{n-1})$. That is
    \begin{align}
    u_n = \left\{
    \begin{array}{ll}
    0, &  \mbox{with probability } P(U_n=0|s_1^N,u_{1}^{n-1})\\ \notag
    1, &  \mbox{with probability } P(U_n=1|s_1^N,u_{1}^{n-1}) \notag
    \end{array}
    \right.
    \end{align}

    \State $n\gets n+1$

  \EndIf

\Until{$n>N$}
\State $x_{1}^{N} \gets u_{1}^{N} \left(G_N\right)^{-1}$
\State The algorithm outputs $x_{1}^{N}$ as the WOM codeword
\end{algorithmic}
\end{algorithm}

The encoding algorithm in Algorithm \ref{polar_encoding_algorithm} takes the inputs including the current memory states of the memory cells $s_1^N$, the high-entropy set ${\mathcal F}_N$, ${\mathcal G}_N$, the to-be-recorded message $v_1^M$, and the freeze bits $f_{1}^{L}$. The algorithm then determines $u_i$ one by one form $i=1$ to $i=N$. If $i\in {\mathcal G}_N$, then $u_i$ is set to one of the freeze bits $f_l$. The freeze bits $f_{1}^{L}$ are some random bits shared by the encoder and decoder. Actually, $f_{1}^{L}$ can be chosen to be deterministic signals without any performance degradation, as in many previous polar coding algorithms. If $i\in {F}_N\cap {\mathcal G}_N^{c}$, then $u_i$ is set to one of the to-be-recorded information bits $v_m$. If $i\notin F_N$, then $u_i$ is randomly set to 1 with probability $P(U_i=1|s_1^N,u_{1}^{i-1})$, and $u_i$ is randomly set to 0 with probability $P(U_i=0|s_1^N,u_{1}^{i-1})$. After all the $u_i$ for $i=1, \ldots, N$ have been determined, a vector $x_{1}^{N}$ is calculated as $x_{1}^{N} = u_{1}^{N} \left(G_N\right)^{-1}$. In other words, $u_{1}^{N} = x_{1}^{N} G_N$. The algorithm finally outputs $x_{1}^{N}$ as the codeword.

As shown in Algorithm \ref{polar_decoding_algorithm}, the decoding algorithm receives channel observation $y_{1}^{N}$, the high-entropy set ${\mathcal F}_N$, ${\mathcal G}_N$, and the freeze bits $f_{1}^{L}$. The algorithm then determines $u_i$ one by one form $i=1$ to $i=N$. If $i\in {\mathcal G}_N$, then $u_i$ is set to one of the freeze bits $f_l$. If $i\in {F}_N\cap {\mathcal G}_N^{c}$, then $u_i$ is decoded using the bit-wise maximal-likelihood estimation. That is,
\begin{align}
u_i = \left\{
\begin{array}{ll}
0,  & \mbox{if } P(U_i=0|y_1^N,u_{1}^{i-1}) > P(U_i=1|y_1^N,u_{1}^{i-1}) \\
1,  & \mbox{otherwise }
\end{array}
\right.
\end{align}
After all the $u_i$ have been determined, the recorded information bits $v_1,v_2,\ldots,v_M$ have already be recovered, because they are just the bits of $u_{1}^{N}$ at the positions in ${\mathcal F}_N$.

\begin{algorithm}
\caption{Joint WOM  and error-control decoding}\label{polar_decoding_algorithm}
\begin{algorithmic}[1]
\State The algorithm takes inputs
\begin{itemize}
\item the channel observations $y_1^N$
\item the high-entropy set ${\mathcal F}_N$
\item the high-entropy set ${\mathcal G}_N$
\item the to-be-decoded message $v_1^M$
\item freeze bits $f_{1}^{L}$
\end{itemize}
\State{$n\gets 1, m\gets 1, l\gets 1$}
\Repeat

    \If{ $n\in {\mathcal G}_N$}

      \State $u_n \gets f_l$
      \State $n\gets n+1$
      \State $l\gets l+1$
  \Else

  \State Set $u_n$ according to the probability distribution $P(U_n|y_1^N,u_{1}^{n-1})$. If $P(U_n=0|y_1^N,u_{1}^{n-1}) > P(U_n=1|y_1^N,u_{1}^{n-1})$, then
         set $u_n=0$. Otherwise, set $u_n=1$.

  \State $n\gets n+1$

  \If{ $n\in {\mathcal F}_N$}

    \State $v_m \gets u_n $

    \State $m\gets m+1$

  \EndIf

  \EndIf

\Until{$n>N$}
\State The algorithm outputs $v_{1}^{M}$ as the decoded message
\end{algorithmic}
\end{algorithm}

\section{Sufficient and Necessary Condition for Less Noisy}

\label{section_less_noisy}

In this section, we prove a sufficient and necessary condition for the noisy reading channel being less noisy than the test channel in data encoding. The sufficient and necessary condition is usually satisfied, unless the flash memory reading (sensing) process is extremely noisy, which may not be the case for most practical systems. As a consequence of the sufficient and necessary condition, we have that ${\mathcal G}_N \subset {\mathcal F}_N$ in Algorithms \ref{polar_encoding_algorithm} and \ref{polar_decoding_algorithm}. Therefore, the WOM coding scheme considered in Section \ref{section_basic_scheme} can be applied for most practical scenarios, despite being a very simple coding scheme.

We consider the joint probability distribution of $S, X, Y$ defined by  Eqs. \ref{model_p_2}, \ref{model_p_3}, and the channel conditional probability distribution $P(Y_i|X_i)$. The probability model is illustrated in Fig. \ref{basic_system_model}, where the probability distribution can be factored as $P(S,X,Z) = P(S)P(X|S)P(Y|X)P(Z|X)$. We may consider a virtual noisy channel, where the random variable $X$ is the input and the random variable $S$ is the output. We call such a channel as the test channel in data encoding and denote it by $X\rightarrow S$. Certainly there exists another noisy channel, where the random variable $X$ is the input and the random variable $Y$ is the output. We call such a channel as the noisy reading channel and denote it by $X\rightarrow Y$. We will prove that the channel $X\rightarrow Y$ is less noisy than the channel $X\rightarrow S$, if and only if $I(X;Y) \geq I(X;S)$.

\begin{figure}[h]
 \centering
 \includegraphics[width=3in]{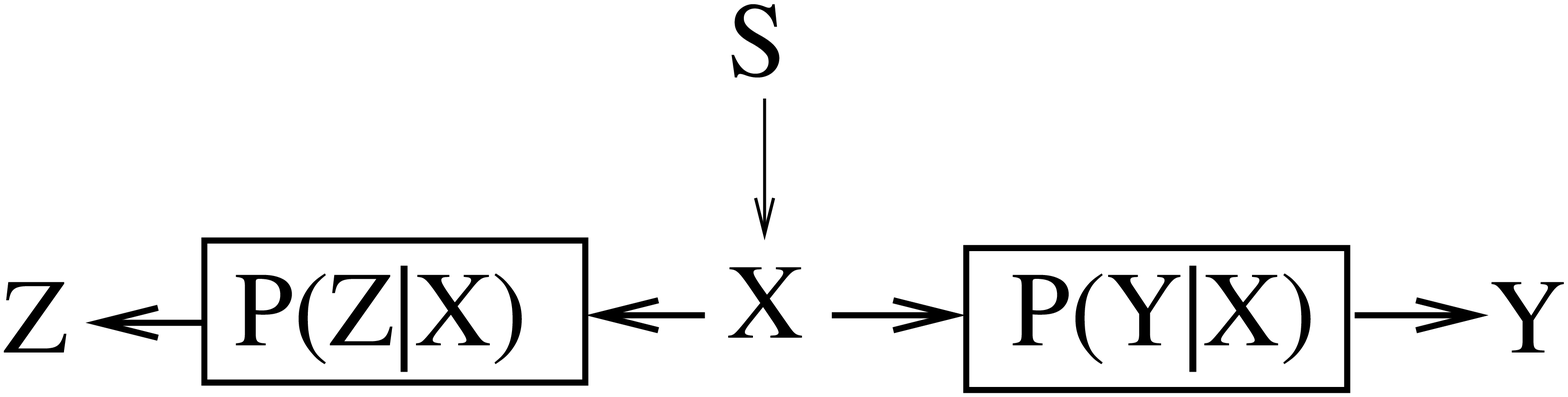}
 \caption{Probability distribution of the considered random variables}
 \label{basic_system_model}
\end{figure}

\begin{definition}
Let ${\mathcal W}: {X} \rightarrow { Y}$ denote the channel with input symbol $X$ and output symbol $Y$. Let ${\mathcal V}: { X} \rightarrow { S}$ denote the channel with input symbol $X$ and output symbol $S$. We say $\mathcal W$ is less noisy than $\mathcal V$ (denoted by ${\mathcal W} \succeq {\mathcal V}$), if $I(Z;Y)\geq I(Z;S)$ for all distribution $P_{Z,X}$, where $Z$ has finite support and $Z\rightarrow X \rightarrow (Y,S)$ form a Markov chain.
\end{definition}

Less noisy has been discussed extensively in the past mainly for the research on broadcast channels. We refer interested readers to \cite{sutter14} and references therein for these previous research.

\begin{theorem}
\label{less_noisy_theorem}
Assume the joint probability distribution of $S,X,Y$ as defined in Eqs. \ref{model_p_2}, \ref{model_p_3}, and the channel conditional probability distribution $P(Y_i|X_i)$. The channel $X\rightarrow Y$ is less noisy than the channel $X\rightarrow S$, if and only if $I(X;Y) \geq I(X;S)$.
\end{theorem}
\begin{proof}
For the necessary part, note that we may let $Z=X$. Then, we have $I(X;Y) \geq I(X;S)$. The necessary part is thus proven.

For the sufficient part, we show that the proof of the theorem can be reduced to the proof of Lemma \ref{less_noisy_lemma}. We have
\begin{align}
& I(Z;S,X) = I(Z;X) + I(Z;S|X) \stackrel{(a)}{=} I(Z;X)  \\
& = I(Z;S) + I(Z;X|S) \\
& I(Z;Y,X) = I(Z;X) + I(Z;Y|X) \stackrel{(b)}{=} I(Z;X)  \\
& = I(Z;Y) + I(Z;X|Y)
\end{align}
where, (a) and (b) follow from the fact that $Z\leftarrow X \leftarrow S,Y$ form a Markov chain.
Thus, all we need to show is
\begin{align}
\label{less_noisy_inequ}
I(Z;Y) - I(Z;S) = I(Z;X|S) - I(Z;X|Y) \geq 0
\end{align}
which follows from Lemma \ref{less_noisy_lemma}.
\end{proof}

\begin{lemma}
\label{less_noisy_lemma}
Assume the joint probability distribution of $S,X,Y$ as defined in Eqs. \ref{model_p_2}, \ref{model_p_3}, and the channel conditional probability distribution $P(Y_i|X_i)$. If $I(X;Y) \geq I(X;S)$, then
\begin{align}
\label{less_noisy_inequ}
I(Z;X|S) - I(Z;X|Y) \geq 0
\end{align}
\end{lemma}

 \begin{figure}[h]
 \centering
 \includegraphics[width=2in]{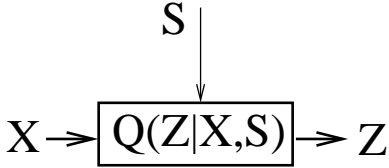}
 \caption{The auxiliary channel model}
 \label{auxiliary_system_model}
\end{figure}

In the proof of Lemma \ref{less_noisy_lemma}, we need to introduce an auxiliary channel model as shown in Fig.
 \ref{auxiliary_system_model}, where
\begin{align}
Q(Z|X,S) = \left\{
\begin{array}{ll}
P(Z|X=0),  & \mbox{for } S=0 \\
P(Z|X),  & \mbox{for } S=1
\end{array}
\right.
\end{align}
Clearly, this channel is a channel with channel state $S$. We will consider a coding scenario, where the state information $S$ is non-causally available to both the encoder and decoder as shown in Fig. \ref{side_information_case}. That is, the side information vector $s{1}^{N}$ is given to both the encoder and decoder before the $N$ channel uses.

\begin{figure}[h]
 \centering
 \includegraphics[width=2in]{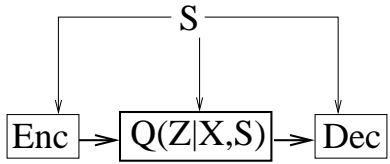}
 \caption{The coding scenario that the channel state information $S$ is non-causally available to both the encoder and decoder}
 \label{side_information_case}
\end{figure}

\begin{lemma}
\label{channel_capacity_lemma}
Assume the joint probability distribution of $S,X,Y$ as defined in Eqs. \ref{model_p_2}, \ref{model_p_3}, and the channel conditional probability distribution $P(Y_i|X_i)$.
Assume that the channel state information $S$ is non-causally available to both the encoder and decoder. Assume that the coding scheme must satisfy a constraint that $\sum_{i=1}^{N}c(s_i,x_i)=N(1-\beta)\gamma(1-\epsilon_N)$, where $N$ is the blocklength of the coding scheme, $\epsilon_N$ goes to zero as $N$ goes to infinity and
\begin{align}
c(s_i,x_i) = \left\{
\begin{array}{ll} 1, & \mbox{if } s_i = 1, x_i =0 \\
0, & \mbox{otherwise }
\end{array}
\right.
\end{align}
Under the above constraint, the capacity (maximal achievable rate) of the channel in Fig. \ref{side_information_case} is equal to $I(Z;X|S)$.
\end{lemma}
\begin{proof}
In the achievable part of the proof, we need to show that the rate $I(Z;X|S)$ is achievable. This part of the proof can be done by using the standard random coding arguments as for example in \cite{cover06}. This part is omitted in this paper due to the space limitation.

In the converse part of the proof, we consider any random codeword vector $X_{1}^{N}$, which achieves the constrained channel capacity as above. Assume $Z_{1}^{N}$ is the corresponding random received channel observation vector. The achievable rate $R$ can be thus bounded as follows.
\begin{align}
NR & \leq  H(X_{1}^{ N} | S_{1}^{N}) = I(Z_{1}^{ N};X_{1}^{ N} | S_{1}^{N}) + P(X_{1}^{ N}|Z_{1}^{ N}, S_{1}^{N}) \notag \\
& \stackrel{(a)}{\leq}  I(Z_{1}^{ N};X_{1}^{ N}| S_{1}^{N}) + 1 + N P_e \notag \\
&=  H(Z_{1}^{N}|S_{1}^{N}) -  H(Z_{1}^{ N}|X_{1}^{ N},S_{1}^{N}) + 1 + NP_e \notag \\
& =  \sum_{i=1}^{ N} H(Z_{i}|Z_{1}^{i-1},S_{1}^{N}) \notag \\
& -  \sum_{i=1}^{ N}H(Z_{i}|Z_{1}^{i-1},X_{1}^{N},S_{1}^{N}) + 1 + NP_e \notag \\
& \stackrel{(b)}{\leq}  \sum_{i=1}^{ N} H(Z_{i}|,S_i) -  \sum_{i=1}^{ N}H(Z_{i}|X_{i},S_{i}) + 1 + NP_e \notag \\
& =  \sum_{i=1}^{ N} I(Z_{i};X_{i}|S_i) + 1 + NP_e
\end{align}
In the above inequality, (a) follows from the Fano's inequality \cite[page 39]{cover06}, where $P_e$ is the decoding error probability; and (b) follows from the fact that conditions decrease entropy and $Z_i$ is independent of all other random variables conditioned on $X_i$.

Now we introduce a random variable $\mathcal{U}$, such that ${\mathcal U}=i$ for $i=1,\ldots, N$ with equal probability. Let $\mathcal{S}$, $\mathcal{X}$, $\mathcal{Z}$ be random variables, such that $\mathcal{S}=S_i$, $\mathcal{X}=X_i$ and $\mathcal{Z}=Z_i$, if $\mathcal{U}=i$. By definition, we have
\begin{align}
\sum_{i=1}^{ N} I(Z_{i};X_{i}|S_i) = N I(\mathcal{Z};\mathcal{X}|\mathcal{S}, \mathcal{U})
\end{align} 
Note that $I(\mathcal{Z};\mathcal{U} | \mathcal{X},\mathcal{S})=0$, therefore
\begin{align}
I(\mathcal{Z};\mathcal{X},\mathcal{U}|\mathcal{S}) & =  I(\mathcal{Z};\mathcal{X}|\mathcal{S})  \\
& = I(\mathcal{Z};\mathcal{U}|\mathcal{S})  + I(\mathcal{Z};\mathcal{X} | \mathcal{U},\mathcal{S})
\end{align}
Because $I(\mathcal{Z};\mathcal{U}|\mathcal{S})\geq 0$, we have
\begin{align}
I(\mathcal{Z};\mathcal{X} | \mathcal{U},\mathcal{S}) \leq  I(\mathcal{Z};\mathcal{X}|\mathcal{S})
\end{align}

The converse part of the proof is done by noting that the random variables $\mathcal{Z}, \mathcal{X},\mathcal{S}$ have exactly the same probability distribution as in Eqs. \ref{model_p_2}, \ref{model_p_3}, and the channel conditional probability distribution $P(Y_i|X_i)$.
\end{proof}

\begin{lemma}
\label{lemma_less_noisy_rate}
Consider the channel with channel state in Fig. \ref{auxiliary_system_model}. Assume the joint probability distribution of $S,X,Y$ as defined in Eqs. \ref{model_p_2}, \ref{model_p_3}, and the channel conditional probability distribution $P(Y_i|X_i)$. Assume that the state information $S_{1}^{N}$ is i.i.d. generated and revealed to both the encoder and decoder before each block of $N$ channel uses. In addition, a random vector $Y_{1}^{N}$ is also i.i.d generated according to the probability distribution $P(Y|S)$. The random vector $Y_{1}^{N}$ is provided to both the encoder and decoder before the block of $N$ channel uses. If $I(X;Y)\geq I(X;S)$, then $I(U;X|Y)$ is an achievable rate for the channel model in Fig. \ref{auxiliary_system_model}.
\end{lemma}
\begin{proof}
We will present a concrete coding schemes, such that the rate $I(Z;X|Y)$ is achieved. We consider a Markov encoding scheme as shown in Fig. \ref{coder_markov}, where the coding scheme uses $K$ blocks of channel uses, and each block contains $N$ channel uses. For each $k$-th channel use block, the transmit message $m$ consists of two parts $V_k$ and $M_k$. The two parts of messages $V_k$ and $M_k$ are encoded using one block encoder into an $N$-dimensional codeword $X_{1}^{N}$. The codeword $X_{1}^{N}$ is then transmitted through the noisy channel using $N$ channel uses.

\begin{figure}[h]
 \centering
 \includegraphics[width=3in]{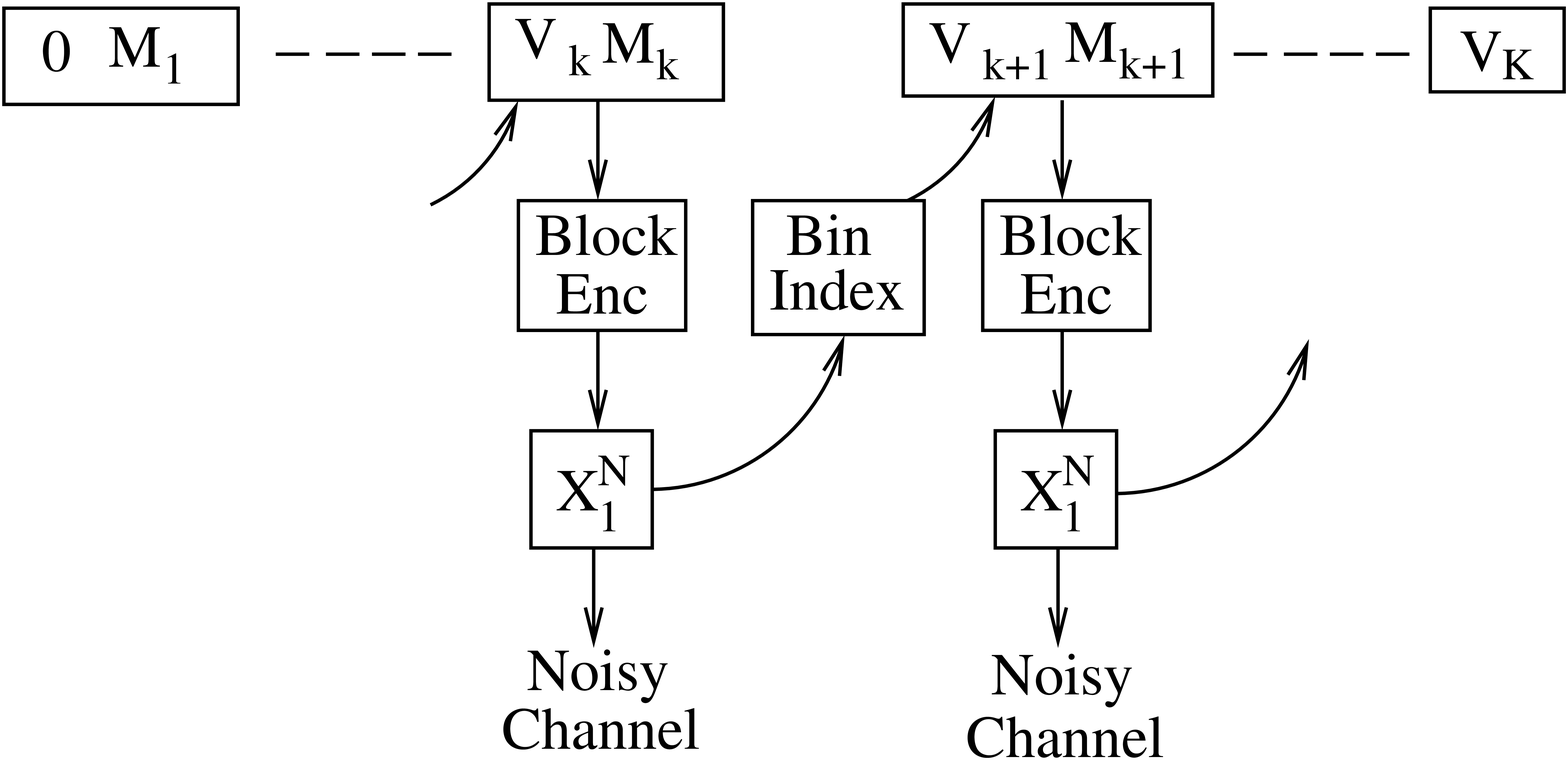}
 \caption{Block Markov coding}
 \label{coder_markov}
\end{figure}

\begin{figure}[h]
 \centering
 \includegraphics[width=3in]{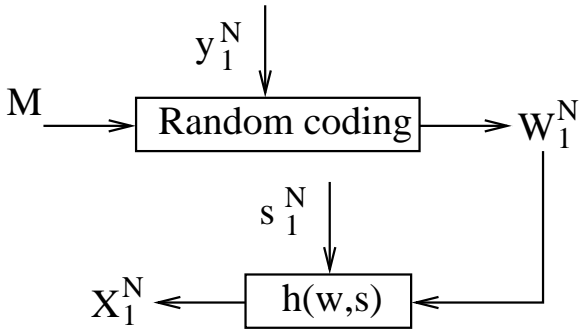}
 \caption{Random coding for each block}
 \label{coder_concrete}
\end{figure}

The details of the block encoding are explained as follows. In the block encoding scheme, an auxiliary codebook is first randomly generated after receiving the random signal $Y_{1}^{N}$. The auxiliary codebook contains $|M|$ codewords, where
\begin{align}
|M|=2^{N(H(X|Y)-3\epsilon_N)},
\end{align}
and $\epsilon_N$ goes to zero as $N$ goes to infinity. Each codeword is randomly generated independent of all other codewords. For each codeword, each codeword bit is randomly generated independent of all other codeword bits according to the probability distribution  $P(X|y_i,s_i=1)$.

After receiving the transmit message $m$, the encoding algorithm selects a codeword $W_{1}^{N}$ from the auxiliary codebook according to the transmit message $m$. The auxiliary codeword $w_{1}^{N}$ is then mapped into a codeword $x_{1}^{N}$ using the following function $h(w,s)$,
\begin{align}
x_i = \left\{
\begin{array}{ll}
0,  & \mbox{if } s_i=0 \\
w_i,  & \mbox{if } s_i=1
\end{array}
\right.
\end{align}
The codeword $x_{1}^{N}$ is then transmitted through the communication channel $Q(Z|X,S)$.

In addition, a bin index $s(x_{1}^{N})$ is generated as follows. We throw all the typical sequences in the typical set $A_{\epsilon_N}^{(N)}(X)$ of $X_{1}^{N}$ into $2^{N(H(X|Y,Z)+2\epsilon_N)}$ bins uniformly in random. The number $s(x_{1}^{N})$ is the index of the bin which contains the transmit codeword $x_{1}^{N}$. We then transmit this bin index during the next block of $N$ channel uses. Essentially, the message $V_{k+1}$ in the next block of channel uses is exactly  the bin index $s(x_{1}^{N})$.

Therefore, for each block of $N$ channel uses, the transmit message $m$ consists of two parts $V_k$, and $M_k$, where $V_k$ is a bin index for the codeword in the previous channel use block, and $M_k$ is the effective transmit message for the current block. The chain of coding is terminated by using one or several blocks to transmit the last bin index reliably (for example using a coding scheme as in Lemma \ref{channel_capacity_lemma}). For the first channel use block, there does not exist a previous bin index. In this case, we set $V_1$ to be an all zero bit string. For each block except the first and last several blocks, the transmit message $m$ contains $\log |M| = N (H(X|Y)-3\epsilon_N)$ bits, the bin index $V_k$ contains $N (H(X|Z,Y)+2\epsilon_N)$ bits, thus the effective transmit message $M_k$ contains approximately $NI(Z;X|Y)$ bits. The transmitting rate of this coding scheme goes to $I(Z;X|Y)$ as both $N$ and $K$ go to infinity.

The decoder may determine the transmitted messages $M$ by working backward by first determining the bin index $V_{k+1}$ and then determining the transmit $V_k,M_k$. Assume that $V_{k+1} = s(X_{1}^{N})$ is decoded correctly. In the step 1, the decoder tries to find a codeword $\hat{X}_{1}^{N}$, such that $\hat{X}_{1}^{N}$ is jointly typical with the channel observation $z_{1}^{N}$ and the received random signal $y_{1}^{N}$ and  $s(\hat{X}_{1}^{N})=s(X_{1}^{N})$.
In the step 2, the decoder tries to find a vector $\hat{W}_{1}^{N}$, such that $\hat{W}_{1}^{N}$ is jointly typical with the estimated codeword $\hat{X}_{1}^{N}$ and the channel state information $s_{1}^{N}$. And finally the decoder maps the codeword $\hat{W}_{1}^{N}$ back to the transmit message $V_k,M_k$. The decoded bin index $V_k$ can thus be used in decoding the previous transmit $V_{k-1},M_{k-1}$.

In the decoding process, there exist the following sources of decoding errors.
\begin{itemize}

\item The received channel state information $S_{1}^{N}$ and $Y_{1}^{N}$ are not typical. We denote this random event by $E_1$.

\item The transmit auxiliary codeword $W_{1}^{N}$ is not a typical sequence. We denote this random event by $E_2$.

\item In the signal transmission process, the channel observation $Z_{1}^{N}$ is not jointly typical with the transmit codeword $X_{1}^{N}$. We denote this random event by $E_3$.
 
\item There exist another typical sequence $\tilde{X}_{1}^{N}$, which has the same cell index as $X_{1}^{N}$. And $\tilde{X}_{1}^{N}$ is jointly typical with $Z_{1}^{N}$ and $Y_{1}^{N}$. We denote this random event by~$E_4$.

\item There exist another auxiliary codeword $\tilde{W}_{1}^{N}$, which is jointly typical with the codeword $X_{1}^{N}$ and channel state information $S_{1}^{N}$. We denote this random event by $E_5$.
\end{itemize}

Using a union bound, we have the decoding error probability
\begin{align}
P_e & \leq  P(E_1) + P(E_2) + P(E_3)  +  P(E_4)  + P(E_5)
\end{align}
In the sequel, we will show that all the 5 terms on the right hand of the above inequality go to zero as the block length $N$ goes to infinity.  It can be easily seen that $P(E_1)$, $P(E_2)$ and $P(E_3)$ go to zeros due to the well-known Asymptotic Equipartition Property (AEP) \cite[Page. 49]{cover06}.

The error probability $P(E_4)$ can be bounded as follows.
\begin{align}
P(E_4) & \stackrel{(a)}{\leq} 1 - \prod_{i=1}^{2^{N(H(X|Z,Y)+\epsilon_N)}}
\left[1-\frac{1}
{2^{N (H(X|Y,Z)+2\epsilon_N)}}\right] \notag \\ 
& \stackrel{(b)}{\leq} 2^{ N (H(X|Z,Y)+\epsilon_N)} 2^{- N (H(X|Z,Y)+2\epsilon_N)}  \notag \\
& \leq 2^{-N \epsilon_{N}}
\end{align}
In the above inequality, (a) follows from the fact that there are $2^{N(H(X|Z,Y)+2\epsilon_N)}$ bins, and there are at most $2^{N(H(X|Z,Y)+\epsilon_N)}$ typical sequences $\tilde{X}_{1}^{N}$, which are jointly typical with $z_{1}^{N}$ and $y_{1}^{N}$. And $(b)$ follows from the inequality $(1-x)^n \geq 1-nx$.
It can be seen clearly that $P(E_4)$ goes to zero as $N$ goes to infinity.

Similarly, we have
\begin{align}
P(E_5) & \stackrel{(a)}{\leq} 1 - \prod_{i=1}^{|M|-1}
\left[1-\frac{2^{N (H(W|X, S)+\epsilon_N)}}
{2^{N (H(W)-\epsilon_N)}}\right] \notag \\
& \stackrel{(b)}{\leq} |M| 2^{-N (I(W;X,S)  - 2\epsilon_{N})} \\  
& \leq \frac{(|M|-1)}{|M|} 2^{N \left(H(X|Y) - I(W;X,S)  - \epsilon_{N}\right)  }
\end{align}
In the above inequality, (a) follows from the fact that there are at least $N (H(W)-\epsilon_N)$ typical sequences $W_{1}^{N}$ and at most $N (H(W|X, S)+\epsilon_N)$ typical sequences $W_{1}^{N}$, which are jointly typical with $x_{1}^{N}$ and $s_{1}^{N}$. And $(b)$ follows from the inequality $(1-x)^n \geq 1-nx$.

Note that
\begin{align}
& I(W;X,S)  = I(W;X|S) + I(W;S) \stackrel{(a)}{=} I(W;X|S) \notag \\
& = H(X|S) - H(X|S,W)  \stackrel{(b)}{=} H(X|S)
\end{align}
where, (a) follows from the fact that the auxiliary codeword bit $W$ is generated independent of the state information $S$, and (b) follows from the fact that $X$ is deterministic given $S$ and $W$.
Because $I(X;Y)\geq I(X;S)$, we have $H(X|Y)\leq H(X|S)$. Then, we have
\begin{align}
H(X|Y) \leq H(X|S) =  I(W;X,S)
\end{align}
Therefore, $P(E_5)$ also goes to zero as $N$ goes to infinity. The lemma follows.
\end{proof}

\begin{proof}[the proof of Lemma \ref{less_noisy_lemma}]
It follows from Lemma \ref{channel_capacity_lemma} that  $I(Z;X|S)$ is the constrained channel capacity for the channel in Fig. \ref{auxiliary_system_model}.  
It follows from Lemma \ref{lemma_less_noisy_rate} that $I(Z;X|Y)$ is an achievable rate for the same channel with an additional shared random signal $Y$. It can be checked that the constraint $\sum_{i=1}^{N}c(s_i,x_i)=N(1-\beta)\gamma(1-\epsilon_N)$ is also satisfied by the coding scheme in Lemma \ref{lemma_less_noisy_rate}.
However, sharing common random signals between the encoder and decoder can not increase channel capacities. Therefore,
\begin{align}
I(Z;X|Y)\leq I(Z;X|S).
\end{align}
The lemma is thus proven. 
\end{proof}

\begin{theorem}
\label{entropy_set_theorem}
Assume the joint probability distribution of $S,X,Y$ as defined in Eqs. \ref{model_p_2}, \ref{model_p_3}, and the channel conditional probability distribution $P(Y_i|X_i)$. Let ${\mathcal G}_N$ and ${\mathcal F}_N$ be the high entropy sets as defined in Section \ref{section_basic_scheme}. Then,
${\mathcal G}_N \subset {\mathcal F}_N$, if and only if $I(X;Y) \geq I(X;S)$. 
\end{theorem}
\begin{proof}
Note that Theorem \ref{entropy_set_theorem} follows immediately from Lemma \ref{less_noisy_lemma}. This is because the following Theorem \ref{universal_theorem} from~\cite{sutter14}.
\end{proof}
\begin{theorem}
\label{universal_theorem}
Let ${\mathcal W}:{\mathcal X}\rightarrow {\mathcal Y}$ and ${\mathcal V}:{\mathcal X}\rightarrow {\mathcal Z}$ be two discrete memory-less channels, such that ${\mathcal W} \succeq {\mathcal V}$, then for any $\epsilon \in (0,1)$, we have ${\mathcal D}_{\epsilon}^N(X|Z) \subset {\mathcal D}_{\epsilon}^{N}(X|Y)$. The definitions of the sets are
\begin{align}
& {\mathcal D}_{\epsilon}^{N} (X|Y) = \{i: H(U_i|U_{1}^{i-1},Y_{1}^{N})\leq \epsilon \} \notag \\
& {\mathcal D}_{\epsilon}^{N} (X|Z) = \{i: H(U_i|U_{1}^{i-1},Z_{1}^{N})\leq \epsilon \} \notag 
\end{align}
\end{theorem}

\section{Conclusion}

\label{sec_conclusion}

The paper discusses a low-complexity joint WOM and error-control coding scheme based on polar coding. We prove a sufficient and necessary condition for the noisy reading channel being less noisy than the test channel in data encoding. We show that ${\mathcal G}_N \subset {\mathcal F}_N$, if and only if $I(X;Y) \geq I(X;S)$. Thus, the above low-complexity coding scheme can be applied for most practical flash memory coding scenarios, unless $I(X;Y) < I(X;S)$, which implies that the reading channels are extremely noisy.

\nocite{*}
\bibliographystyle{IEEEtran}
\bibliography{the_bib}

\end{document}